\documentclass[10pt,conference]{IEEEtran}

\usepackage{amsmath,amssymb,amsthm}
\usepackage{algorithmic}
\usepackage{graphicx}
\usepackage{hyperref}
\usepackage{cite}
\usepackage{quantikz}
\usepackage{booktabs}
\usetikzlibrary{backgrounds,decorations.pathreplacing}

\newtheorem{theorem}{Theorem}
\newtheorem{lemma}[theorem]{Lemma}
\newtheorem{proposition}[theorem]{Proposition}
\newtheorem{corollary}[theorem]{Corollary}
\newtheorem{definition}{Definition}

\begin{document}

\title{Coherent Rollout Oracles for\\ Finite-Horizon Sequential Decision Problems}

\author{\IEEEauthorblockN{Nishant Shukla}
\IEEEauthorblockA{Independent Researcher}}

\maketitle

\begin{abstract}
Coherent quantum rollout for sequential decision problems requires a
unitary simulator: randomness must live in explicit quantum registers,
and basis-state selectors must be mapped to actions reversibly. With
branch-dependent valid actions, this mapping is \emph{totalized
coherent rank-select} over an entangled $N$-bit validity mask: return
the position of the $r$-th valid bit, or a sentinel if $r$ is out of
range.

We give the first reversible-circuit complexity analysis of this
primitive. For selector width $w = \lceil \log_2(N+1) \rceil$,
rank-select admits an $O(Nw)$-gate low-ancilla bounded-span scan,
proved gate-optimal in its model, and an $O(N\log w)$-gate low-ancilla
blocked construction when long-range gates are available; across all
bounded-fan-in layouts, the unconditional gate lower bound is
$\Omega(N)$. Composing rank-select with reversible transition and
predicate-evaluation circuits gives an explicit polynomial-size
coherent rollout oracle for finite-horizon planning problems
satisfying these primitive assumptions.

The resulting oracle satisfies the access model of the best-arm
pipeline of Wang et al.~\cite{wang2021quantum}, yielding
$\widetilde{O}(\sqrt{k}/\varepsilon)$ coherent oracle calls against
the standard classical $\Omega(k/\varepsilon^2)$ arm-pull lower bound.
We give a bounded-influence lifting theorem that extends this lower-bound
construction from a base configuration to an exponential family of
configurations. We instantiate the construction on SIR epidemic
intervention, with a stochastic placement-game sanity check, and
machine-check the main results in Lean~4. Code and proofs:
\url{https://github.com/BinRoot/b01t/tree/main/demos/rollout}.
\end{abstract}


% ============================================================
\section{Introduction}
\label{sec:intro}
% ============================================================

In a classical rollout simulator, action choices are typically drawn
from an internal random-number generator hidden from the caller. A
coherent quantum rollout cannot use such implicit randomness. Its
random choices must live in explicit registers so the rollout step is
unitary and invertible. A total rollout
oracle must turn a selector (a basis-state index naming the rollout's
random choice) into a valid action without measurement or rejection.
When action validity is a bitstring, the selector-to-position map is
totalized coherent \emph{rank-select}. Given the current state and a rank
$r$, the oracle returns the position of the $r$-th valid action, or a
sentinel when $r$ is out of range.

Totalized coherent rank-select is the form the oracularization
barrier of Dunjko et
al.~\cite{dunjko2016quantum} takes for rollout planning under
branch-dependent valid actions. Resolving rank-select, together with
reversible stochastic transition and terminal evaluation, yields a
polynomial-size coherent rollout oracle. Composing the resulting
oracle with the best-arm pipeline of Wang et al.\ gives a
near-quadratic query speedup for the induced planning problem.

Existing approaches split into two camps: assume abstract oracle access
whose implementation cost is left outside the
model~\cite{wang2021quantum,wang2025quantum,montanaro2015quantum},
or require explicit state
enumeration~\cite{ronagh2019quantum,luo2025quantum}. Neither covers the
implicit state spaces where rollout-based planning is the standard
classical tool.

We give a constructive normal form for rollout planning whenever
actions are selected under branch-dependent validity, randomness
lives in registers, transitions are polynomial-size
reversible circuits, and terminal rewards are polynomial-size
predicates. Within this regime, the construction realizes the
coherent rollout oracle required by quantum best-arm pipelines
without assuming abstract environment access.

Our contributions are:
\begin{itemize}
  \item \emph{A coherent rank-select decoder for branch-dependent
    valid actions (Section~\ref{sec:oracle:index}).} We give two
    coherent rank-select constructions, both using $O(w)$ clean
    ancillae: a sequential scan with $O(Nw)$ gates that is optimal
    in the bounded-span model (matching a total-prefix-span lower
    bound), and a blocked construction with $O(N \log w)$ gates that
    exploits long-range connectivity (Theorem~\ref{thm:rs-block}).
    The unconditional gate lower bound across all layouts is
    $\Omega(N)$, and closing the gap to $O(N \log w)$ is open.
    To our knowledge, this is the first reversible-circuit complexity
    analysis of coherent rank-select under branch-dependent validity.
  \item \emph{An explicit polynomial-size coherent rollout oracle
    for implicit-state finite-horizon planning
    (Section~\ref{sec:oracle}).} We decompose the oracle into three
    reversible phases (rank-select indexing, stochastic transition,
    terminal evaluation) that compose into a polynomial-size unitary
    (Theorem~\ref{thm:oracle}), with gate-count and qubit-count
    formulas per call. The oracle is implemented in
    Qiskit with branchwise agreement against classical rollouts on
    all tested instances (Section~\ref{sec:sway}).
  \item \emph{Rollout speedup across exponentially many
    configurations.} Our oracle turns the known best-arm speedup of
    Wang et al.~\cite{wang2021quantum} into an explicit rollout
    algorithm for this planning class, using
    $\widetilde{O}(\sqrt{k}/\varepsilon)$ coherent calls
    (Corollary~\ref{thm:upper}) with per-call circuit costs.
    We also give a bounded-influence lifting theorem
    (Theorem~\ref{thm:lifting}) that extends the standard classical
    $\Omega(k/\varepsilon^2)$ hard construction
    (Proposition~\ref{thm:lower}) from a base configuration to
    exponentially many locally coupled configurations.
    Appendix~\ref{app:decay} shows the lifting applies to local
    subcritical spreading dynamics, and
    Section~\ref{sec:epidemic} instantiates the result on SIR
    epidemic intervention.
\end{itemize}
\smallskip\noindent\textbf{Artifact summary.}
The Lean~4 artifact machine-checks the rank-select bounds, three-phase
oracle cost formulas, classical $\Omega(k/\varepsilon^2)$ bound, and
bounded-influence lifting, with the quantum
$\widetilde{O}(\sqrt{k}/\varepsilon)$ wrapper checked modulo cited
IQAE~\cite{grinko2021iterative} and D\"urr--H\o yer~\cite{durr1996quantum} bounds. The Qiskit artifact implements
rank-select and the three-phase oracle with branchwise validation
against classical rollouts on all tested instances, and
Appendix~\ref{app:artifact} maps theorems to modules. For
Theorem~\ref{thm:oracle}, Lean covers the cost composition and the
Qiskit branchwise tests validate unitarity and ancilla cleanness.

Section~\ref{sec:prelim} reviews the quantum best-arm primitives
used later. Section~\ref{sec:oracle} constructs the oracle.
Section~\ref{sec:lower} proves the classical lower bound.
Section~\ref{sec:upper} gives the quantum upper bound.
Section~\ref{sec:sway} demonstrates the pipeline on an epidemic
intervention model, with a stochastic placement-game sanity check.
Section~\ref{sec:conclusion} concludes.

% ============================================================
\section{Preliminaries}
\label{sec:prelim}
% ============================================================

Our algorithm composes two known primitives: amplitude estimation to
evaluate individual actions, and quantum maximum finding to identify the
best one (Figure~\ref{fig:full-system}). This section defines the
selection problem and reviews both.

\subsection{The Selection Problem}
\label{sec:prelim:selection}

Given $k$ actions whose values are revealed only through stochastic
simulation, the goal is to identify one within $\varepsilon$ of the best
with high confidence. Each action $i$ has an unknown expected value
$v_i$, also called its arm mean in the bandit literature, and a
single rollout returns a noisy observation centered around $v_i$.

\begin{definition}\label{def:eps-correct}
Given $k$ candidate actions with expected values $(v_i)_{i=1}^{k}$, an
algorithm is $(\varepsilon, 2/3)$-correct if its output $\hat{a}$ satisfies
\begin{equation}\label{eq:eps-correct}
  \Pr\!\bigl[\,v_{\hat{a}} \geq \max_j v_j - \varepsilon\,\bigr]
  \;\geq\; \frac{2}{3}\,.
\end{equation}
The threshold $2/3$ is the standard convention~\cite{evendar2006action}
and is boostable to any $1{-}\delta$ by majority vote.
\end{definition}

We measure algorithms by the number of rollout-oracle calls needed to
achieve this guarantee. Section~\ref{sec:lower} gives the corresponding
classical $\Omega(k/\varepsilon^2)$ arm-pull lower bound, and
Section~\ref{sec:upper} applies Wang et al.'s~\cite{wang2021quantum}
quantum upper bound to the oracle constructed in
Section~\ref{sec:oracle}.

\subsection{Coherent Quantum Amplitude Estimation}
\label{sec:prelim:ae}

Quantum amplitude estimation~\cite{brassard2002quantum} estimates
the probability that a designated success qubit of $U$ equals
$\lvert 1\rangle$, to additive error~$\varepsilon$, using
$\widetilde{O}(1/\varepsilon)$ applications of $U$ and $U^\dagger$. The procedure is a single coherent circuit of
controlled powers of a Grover iterate built from $U$, which is why
Section~\ref{sec:oracle} constructs the rollout oracle as an explicit
reversible unitary.

Amplitude estimation requires \emph{unitary} oracle access: the ability
to apply both $U$ and $U^\dagger$. This matches the coherent oracle
model of Wang et al.~\cite{wang2021quantum}, equivalently the strong
quantum oracle of Wang et al.~\cite{wang2025quantum}. The quadratic
speedup depends on being able to query the arms and the randomness of
their rewards in superposition~\cite{buchholz2025quantum}. Under the
weaker channel model, where the oracle's internal randomness is
resampled independently each call, the quadratic speedup no longer
applies~\cite{wang2025quantum,buchholz2025quantum}. Our oracle
provides coherent amplitude access because its randomness lives in
quantum registers we prepare ourselves;
Section~\ref{sec:oracle:transition} gives the construction.

\subsection{Quantum Maximum Finding}
\label{sec:prelim:maxfind}

The final ingredient is quantum maximum
finding~\cite{durr1996quantum}, a Grover-based~\cite{grover1996fast}
search over $k$ actions in $O(\sqrt{k})$ comparisons. Wang et al.~\cite{wang2021quantum}
compose maximum finding with amplitude estimation as the standard
baseline for quantum best-arm identification, and Wang and
Lui~\cite{wang2025quantum} give tight query bounds under oracle access.

\begin{figure}[t]
\centering
\begin{tikzpicture}[
  lbl/.style={font=\footnotesize},
  cost/.style={font=\footnotesize, text=gray!70!black},
  arr/.style={->, >=stealth, line width=0.6mm},
]
  % ── Input: k moves (text only) ──
  \node[lbl, font=\footnotesize\bfseries] at (0, -0.9) {$k$ moves};

  % arrow to Level 1
  \draw[arr] (0.55, -0.9) -- (0.95, -0.9);

  % ── Level 1: Max-Finding ──
  \draw[rounded corners=4pt, line width=0.4mm, gray!60]
    (1.0, 0.0) rectangle (5.6, -2.8);
  \node[lbl, font=\footnotesize\bfseries, anchor=west]
    at (1.15, -0.25) {Max-Finding};
  \node[cost, anchor=east] at (5.45, -0.25) {$O(\sqrt{k})$};

  % ── Level 2: Coherent Amplitude Estimation ──
  \draw[rounded corners=3pt, line width=0.5mm, gray!60]
    (1.25, -0.55) rectangle (5.35, -2.55);
  \node[lbl, font=\footnotesize\bfseries, anchor=west]
    at (1.4, -0.8) {Amplitude Est.};
  \node[cost, anchor=east] at (5.2, -0.8) {$\widetilde{O}(1/\varepsilon)$};

  % ── Level 3: Rollout Oracle (single block) ──
  \draw[fill=gray!10, rounded corners=2pt, line width=0.7mm]
    (1.55, -1.15) rectangle (5.05, -2.25);
  \node[lbl, font=\footnotesize\bfseries] at (3.3, -1.5)
    {Rollout Oracle};
  \node[cost, font=\scriptsize] at (3.3, -1.85)
    {$\mathrm{poly}(N, H, w)$ gates};

  % ── Output: ε-optimal (text only) ──
  \node[lbl, font=\footnotesize\bfseries] at (6.8, -0.9) {$\varepsilon$-optimal};

  % arrow from Level 1 to output
  \draw[arr] (5.65, -0.9) -- (6.1, -0.9);

  % ── Bottom: total cost ──
  \node[font=\footnotesize\bfseries] at (3.3, -3.15)
    {Total: $\widetilde{O}(\sqrt{k}/\varepsilon)$ oracle calls};
\end{tikzpicture}
\caption{The quantum move-selection algorithm as three nested levels.
The outer loop (D\"urr--H\o yer) makes $O(\sqrt{k})$
comparisons; each comparison uses amplitude estimation with
$\widetilde{O}(1/\varepsilon)$
calls to the rollout oracle from Section~\ref{sec:oracle}. Total:
$\widetilde{O}(\sqrt{k}/\varepsilon)$ oracle calls at
$\mathrm{poly}(N, H, w)$ gates each.}
\label{fig:full-system}
\end{figure}
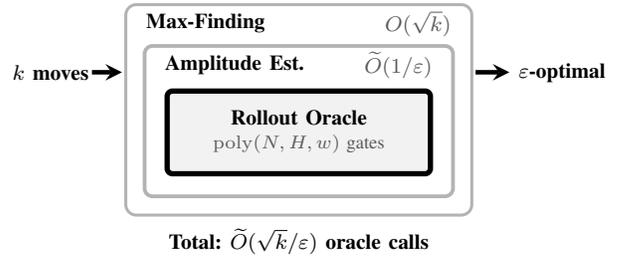

% ============================================================
\section{A Normal Form for Coherent Rollout Oracles}
\label{sec:oracle}
% ============================================================

To provide the innermost level used by amplitude estimation in
Figure~\ref{fig:full-system}, we need an explicit unitary circuit that
simulates a complete rollout. This section constructs one by decomposing
it into three reversible phases that compose into a single oracle.

\subsection{Phase 1: Coherent Rank-Select Indexing}
\label{sec:oracle:index}

Which actions are valid depends on the current state, and in a quantum
rollout the state is in superposition, so the validity predicate is
branch-dependent, and rank-select must produce the correct
result in every branch simultaneously. This is a quantum analogue
of the classical \emph{select} operation on
bitvectors~\cite{jacobson1989space}: given a validity mask, return the
position of the $r$-th set bit, but here the mask is entangled with
the computation rather than fixed.

Formally, let $M_s \in \{0,1\}^N$ be the validity mask in
state~$s$, with $M_s(i)$ denoting whether position
$i \in \{0, \ldots, N{-}1\}$ is valid. Rank-select outputs either
a valid position in $\{0, \ldots, N{-}1\}$ or the sentinel value~$N$,
so the selector, counter, and output-index registers all have width
$w := \lceil\log_2(N{+}1)\rceil$. Given a state register
$\lvert s\rangle$, a rank register $\lvert r\rangle$, and a clean
output register $\lvert 0\rangle^{\otimes w}$, the primitive implements
the total map
\begin{equation}\label{eq:rank-select}
  \lvert s\rangle\lvert r\rangle\lvert 0\rangle^{\otimes w}
  \;\mapsto\;
  \lvert s\rangle\lvert r\rangle\lvert j_s(r)\rangle\,,
\end{equation}
where $j_s(r)$ is the 0-indexed position of the $r$-th 1-bit in $M_s$
when $0 \leq r < \lvert M_s\rvert_1$, and $j_s(r) := N$ otherwise
(sentinel), binary-encoded on $w$~qubits, matching the
compare-and-increment circuit of Figure~\ref{fig:rank-selector}
(rank register initialized to $\ket{0}$). E.g., with $N = 8$ and
$M_s = 01101000$, valid positions are $\{1, 2, 4\}$, so $j_s(0) = 1$,
$j_s(1) = 2$, $j_s(2) = 4$, and $j_s(r) = 8$ for $r \geq 3$.

The sentinel value~$N$ denotes ``no valid action selected''; downstream
transition circuits treat it as a designated no-op branch. The output
register starts clean; the scan first initializes it to the sentinel
value~$N$. If the unique matching position $i$ is found, the
corresponding scan cell replaces the sentinel by $i$. If no match is
found, the sentinel remains. The selector register is
prepared uniformly over all $2^w$ bit strings, so the induced policy
is not uniform over valid actions. In branches with
$\lvert M_s\rvert_1 < 2^w$, the sentinel receives the aliased mass
(Section~\ref{sec:oracle:conditions}).

\smallskip\noindent\textbf{Circuit models.}
The rank-select upper bound is an explicit reversible circuit, but
the lower bounds depend on what kinds of circuits are allowed. We
therefore use two models. The \emph{sequential-scan selector} scans
the $N$-bit validity register left to right, updating a workspace
register with state space $\mathcal{W}$, where $W = |\mathcal{W}|$,
via a reversible step function at each position. Its output is
determined by the final workspace state. The \emph{bounded-fan-in reversible circuit}
is an unrestricted DAG of gates, each acting on at most
$\kappa$~qubits (e.g., $\kappa = 3$ for $\{$CNOT, Toffoli$\}$). After
$G$~gates, each output qubit depends on at most $w + \kappa G$ input
bits, where $w$ is the output width (backward light cone). The
unrestricted model gives the layout-free $\Omega(N)$ input-dependence
bound, against which a blocked construction
(Theorem~\ref{thm:rs-block}) achieves $O(N \log w)$ gates. The
bounded-span scan model captures our reference implementation and
gives the matching $\Omega(Nw)$ gate lower bound.

The rest of this subsection explains why the scan carries the right
information. First, any bounded-fan-in implementation must inspect
linearly many validity bits, giving an unconditional $\Omega(N)$ gate
lower bound. Second, any left-to-right scan must carry enough state
past each prefix to determine how many valid positions have already
been seen. We prove this as a prefix/suffix communication lower
bound. The resulting crossing-gate and workspace bounds show that the
$O(Nw)$ reversible scan is tight in the bounded-span scan model.

\begin{proposition}[Unconditional gate lower bound]\label{prop:rs-gates-lb}
Any bounded-fan-in reversible circuit (fan-in at most $\kappa$)
computing the map of equation~\eqref{eq:rank-select} and returning
ancillae to $\lvert 0\rangle$ requires $\Omega(N)$~gates.
\end{proposition}

\begin{proof}
For each input position $p \in \{0, \ldots, N{-}1\}$, the singleton
mask $M = e_p$ (only bit $p$ set) at rank $r = 0$ has $j_s(0) = p$,
while the all-zero mask at $r = 0$ has $j_s(0) = N$. The outputs
differ, so the rank-select map depends on every input mask bit. The
backward light cone of an output qubit grows by at most $\kappa - 1$
new input qubits per gate, so the union of the $w$ output qubits'
light cones covers at most $w + \kappa G$ input qubits after $G$
gates. Hence $w + \kappa G \geq N$, giving $\kappa G \geq N - w =
\Omega(N)$ for $w = O(\log N)$ and constant $\kappa$.
\end{proof}

The main lower-bound tool is a prefix/suffix communication game on a
canonical hard family. For cut position $t \in \{1,\dots,N{-}1\}$, let
\[
  S_t \;:=\; \{\, w \in \mathbb{Z} : \max(0,\, 2t - N) \le w \le t - 1 \,\},
\]
with $|S_t| = \min(t,\, N{-}t)$, and for each $w \in S_t$ let
$P_w \in \{0,1\}^t$ be any prefix of Hamming weight (popcount)~$w$. Writing
$1^{N-t}$ for the all-ones string of length $N{-}t$ and $\Vert$ for
bitstring concatenation, the canonical mask is
$M^{(w)} := P_w \Vert 1^{N-t}$. Since $w + (N{-}t) \ge t$, $M^{(w)}$
has at least $t$ ones, and at rank $r = t{-}1$ the selected position
lies in the suffix at $j_s(r) = 2t - w - 1$.

\begin{lemma}[Rank-select cut lower bound]\label{thm:rs-cut}
For any cut position $t \in \{1, \dots, N{-}1\}$, any deterministic
protocol in which Alice holds the prefix $M[0..t{-}1]$, Bob holds the
suffix $M[t..N{-}1]$ together with the rank $r$, and Bob produces
$j_s(r)$, must communicate at least
$\lceil \log_2 \min(t, N{-}t) \rceil$ bits.
\end{lemma}

\begin{proof}
Fix Bob's suffix to $1^{N-t}$ and $r = t{-}1$, and vary Alice's prefix
over $\{P_w : w \in S_t\}$. Bob's input is identical across the
family, so two prefixes yielding the same transcript would yield the
same output. But the outputs $2t - w - 1$ are pairwise distinct
on $S_t$. Hence the reachable transcripts number at least
$|S_t| = \min(t, N{-}t)$.
\end{proof}

\begin{lemma}[Crossing-gate lower bound]\label{cor:rs-crossing}
Let $C$ be a bounded-fan-in reversible circuit (fan-in at most
$\kappa$) computing the map of equation~\eqref{eq:rank-select}, and
let the wires be partitioned into $L \cup R$ for some cut position~$t$
so that $M[0..t{-}1]$ lies on $L$, $M[t..N{-}1]$ and $r$ lie on $R$, and
the output register lies on $R$. Ancillae may be placed on either side.
If $G_\times$ is the number of gates whose support touches both $L$
and $R$, then
\[
  G_\times \;\ge\; \frac{1}{2\kappa}\,\lceil \log_2 \min(t, N{-}t) \rceil.
\]
\end{lemma}

\begin{proof}
Alice simulates every gate supported entirely in~$L$ and Bob every
gate supported entirely in~$R$. For a crossing gate with $a$ wires
on Alice's side (and $b \le \kappa - a$ on Bob's side), Alice sends
those $a$ bits to Bob, Bob applies the gate, and Bob returns the
updated $a$ bits if Alice needs those wires for subsequent gates:
at most $2a \le 2\kappa$ communicated bits per crossing gate.
Ancillae carry no input information (they start and end at
$\lvert 0\rangle$), so placing them on either side does not help
either party. The resulting protocol computes Bob's output using at
most $2\kappa\, G_\times$ bits, which Lemma~\ref{thm:rs-cut} forces
to be at least $\lceil \log_2 \min(t, N{-}t) \rceil$.
\end{proof}

Write $\mathrm{pc}_s(p) := \sum_{i \le p} M_s(i)$ for the prefix
count of valid positions in mask~$M_s$ up through position~$p$.

\begin{corollary}[Sequential scan tracks prefix rank at every cut]\label{thm:rs-emergence}
Let $W_t$ denote the workspace state of a sequential-scan selector
after processing positions $0, \dots, t{-}1$ and before processing
position~$t$. For every cut position $t \in \{1, \dots, N{-}1\}$, the
map $w \mapsto W_t$ on the canonical family
$\{M^{(w)} : w \in S_t\}$ is injective, so the workspace register
has at least $\min(t, N{-}t)$ reachable states and therefore carries
at least $\lceil \log_2 \min(t, N{-}t) \rceil$ bits. On this family
the injective parameter is the prefix weight
$w = \mathrm{pc}_{M^{(w)}}(t{-}1)$, so the workspace at the cut
decodes to the processed-prefix count.
\end{corollary}

\begin{proof}
The workspace $W_t$ is the only channel across the time cut from the
processed prefix to the future computation. Alice runs the scan on
her prefix, sends $W_t$ to Bob, and Bob continues the scan;
Lemma~\ref{thm:rs-cut} forces this protocol to use at least
$\lceil \log_2 \min(t, N{-}t) \rceil$ bits, so the workspace has
$\ge \min(t, N{-}t)$ reachable states. The $\min(t, N{-}t)$ canonical
inputs $\{M^{(w)} : w \in S_t\}$ therefore map injectively into
$W_t$, with injective parameter the processed-prefix count
$w = \mathrm{pc}_{M^{(w)}}(t{-}1)$.
\end{proof}

Setting $t = \lfloor N/2 \rfloor$ in Corollary~\ref{thm:rs-emergence}
gives $\min(t, N{-}t) = \lfloor N/2 \rfloor$, so any sequential-scan
selector needs a workspace of $\Omega(\log N)$ bits. We call a
circuit \emph{bounded-span} with span $\sigma$ if each gate's support
crosses at most $\sigma$ prefix/suffix cuts in the ordered mask
layout. Summing Lemma~\ref{cor:rs-crossing} over prefix cuts
gives a matching $\Omega(Nw)$ gate lower bound in this model
(Theorem~\ref{cor:rs-tight} below).

A matching sequential-scan construction realizes these bounds.
The circuit scans all $N$~positions, maintaining a running rank
counter. At each position, it checks whether the rank matches the
desired index and the position is valid; if both hold, it flags that
position as selected. After the scan, the rank counter is uncomputed
and reused (Figure~\ref{fig:rank-selector}). The cost is $O(Nw)$~gates
plus the cost of evaluating the validity predicate at each position,
and $O(w)$~ancilla qubits, where $w = \lceil \log_2(N{+}1) \rceil$ is
the counter width.

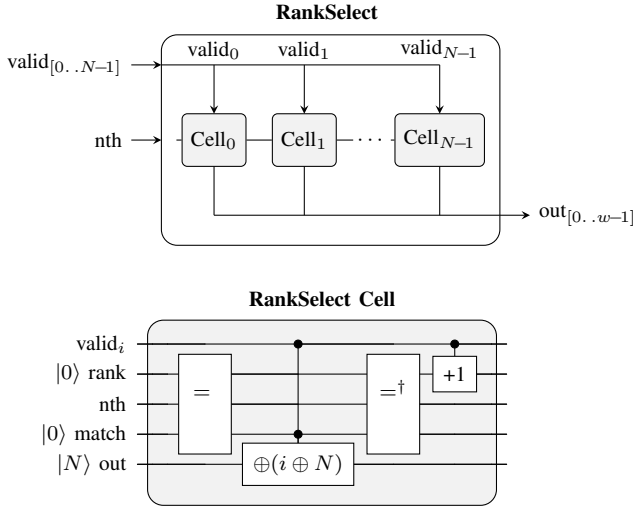
\begin{figure}[!t]

% ── Figure 2a: the full RankSelect composed from N cells ──
\makebox[\columnwidth][c]{\begin{tikzpicture}[
  cell/.style={draw, fill=gray!10, minimum width=0.7cm, minimum height=0.7cm,
               rounded corners=2pt, font=\footnotesize},
  lbl/.style={font=\footnotesize},
  arr/.style={->, >=stealth},
]
  % Cells
  \node[cell] (c0) at (0,0) {Cell$_0$};
  \node[cell] (c1) at (1.2,0) {Cell$_1$};
  \node[lbl]  (dots) at (2.1,0) {$\cdots$};
  \node[cell] (cn) at (3.0,0) {Cell$_{N\!-\!1}$};

  % Threading wire: nth through cells
  \draw (-0.5, 0) -- (c0.west);
  \draw (c0.east) -- (c1.west);
  \draw (c1.east) -- (1.85, 0);
  \draw (2.35, 0) -- (cn.west);
  \draw (cn.east) -- (3.6, 0);

  % Per-cell inputs from top (v_i fan-out)
  \foreach \x/\lab in {0/$v_0$, 1.2/$v_1$, 3.0/$v_{N\!-\!1}$} {
    \draw[arr] (\x, 1.0) -- (\x, 0.35);
  }
  \node[lbl] at (0, 1.2) {$\text{valid}_0$};
  \node[lbl] at (1.2, 1.2) {$\text{valid}_1$};
  \node[lbl] at (3.0, 1.2) {$\text{valid}_{N\!-\!1}$};

  % Per-cell outputs drop to gather bus (no arrowheads)
  \foreach \x in {0, 1.2, 3.0} {
    \draw (\x, -0.35) -- (\x, -1.0);
  }

  % Encapsulating RankSelect box
  \draw[rounded corners=6pt] (-0.7,-1.4) rectangle (3.8,1.4);
  \node[font=\footnotesize\bfseries] at (1.5,1.7) {RankSelect};

  % - External inputs on the left -
  % v[0..N-1]: short arrow into left side, then fan-out line across top
  \draw[arr] (-1.1, 1.0) -- (-0.7, 1.0);
  \node[lbl, left] at (-1.1, 1.0) {$\text{valid}_{[0\mathinner{.\,.}N\!-\!1]}$};
  % horizontal fan-out bus along top inside box
  \draw (-0.7, 1.0) -- (3.0, 1.0);

  % nth: short arrow into left side at the threading wire level
  \draw[arr] (-1.1, 0) -- (-0.7, 0);
  \node[lbl, left] at (-1.1, 0) {nth};

  % - External output on the right -
  % shared w-qubit output bus along bottom, then short arrow out right
  \draw (0, -1.0) -- (3.8, -1.0);
  \draw[arr] (3.8, -1.0) -- (4.2, -1.0);
  \node[lbl, right] at (4.2, -1.0) {$\text{out}_{[0\mathinner{.\,.}w\!-\!1]}$};
\end{tikzpicture}}

\vspace{14pt}

% ── Figure 2b: one cell detail ──
\tikzset{cellbox/.style={rounded corners=6pt,draw,thin,
  append after command={
    \pgfextra{\begin{pgfonlayer}{background}
      \fill[gray!10,rounded corners=6pt]
        (\tikzlastnode.south west) rectangle (\tikzlastnode.north east);
    \end{pgfonlayer}}
  }
}}
\tikzset{celllbl/.style={label position=above,anchor=south,yshift=-0.15cm,font=\footnotesize\bfseries}}
\makebox[\columnwidth][c]{\begin{quantikz}[font=\footnotesize, row sep={0.4cm,between origins}, column sep=0.2cm, thin lines]
  \lstick{$\text{valid}_i$}          & \qw & \qw \gategroup[5,steps=5,style=cellbox,label style=celllbl]{RankSelect Cell} & \qw & \ctrl{4} & \qw & \ctrl{1} & \qw & \qw \rstick{\phantom{$\text{valid}_i$}} \\
  \lstick{$\ket{0}$ rank}            & \qw & \qw & \gate[3,style={minimum width=0.7cm},label style={yshift=0.2cm}]{=\vphantom{\ensuremath{^\dag}}}   & \qw      & \gate[3,style={minimum width=0.7cm},label style={yshift=0.2cm}]{=^\dag} & \gate[1]{\text{+1}} & \qw & \qw \rstick{\phantom{$\ket{0}$ rank}} \\
  \lstick{nth}                       & \qw & \qw &                                          & \qw      &                                             & \qw & \qw & \qw \rstick{\phantom{nth}} \\
  \lstick{$\ket{0}$ match}           & \qw & \qw &                                          & \ctrl{1} &                                             & \qw & \qw & \qw \rstick{\phantom{$\ket{0}$ match}} \\
  \lstick{$\ket{N}$ out} & \qw & \qw & \qw                                      & \gate{\oplus (i \oplus N)} & \qw                                  & \qw & \qw & \qw \rstick{\phantom{$\ket{N}$ out}}
\end{quantikz}}

\caption{\textbf{Top:}~the RankSelect circuit chains $N$~per-position
blocks (``Cell$_i$''), taking inputs
$\text{valid}_{[0\mathinner{.\,.}N\!-\!1]}$ and nth, and producing
a binary-encoded position index in a shared $w$-qubit output
register $\text{out}_{[0\mathinner{.\,.}w\!-\!1]}$.
\textbf{Bottom:}~detail of one block, applied at position~$i$.
The output register is preloaded with the binary encoding of the
sentinel value~$N$. On a match at position $i$ (rank $=$ nth and
$\text{valid}_i = 1$), the circuit replaces the sentinel by $i$ in
$\text{out}$ by XORing the constant $i \oplus N$ into $\text{out}$,
using $O(w)$ multi-target CNOTs on its 1-bits. Rank increments on each
valid position.}
\label{fig:rank-selector}
\end{figure}

\begin{theorem}[Sequential-Scan Rank-Select: Upper Bound]\label{thm:rank-select}
Given an $N$-qubit validity register (possibly entangled with the
computation) and a $w$-qubit index register, the sequential-scan
rank-select circuit selects the $r$-th valid item using $O(Nw)$~gates, $O(Nw)$~depth
(the scan is inherently sequential, and the ripple-carry increment
contributes $O(w)$~depth per step), and $O(w)$~ancilla qubits, with
all ancillae fully uncomputed and reusable.
\end{theorem}

\begin{proof}
At each scan position, the circuit performs a $w$-bit comparison, an
$O(w)$-controlled output-update stage, uncomputes the comparison
scratch, and conditionally increments the $w$-bit prefix counter via
a ripple-carry adder~\cite{cuccaro2004new}. Each stage costs $O(w)$ bounded-fan-in
reversible gates, so each scan position costs $O(w)$ gates and the
total over $N$ positions is $O(Nw)$. The steps are sequential, and
within each step the ripple-carry chain gives $O(w)$ depth, so total
depth is $O(Nw)$.
\end{proof}

A decoder applies the action at the indexed position (e.g., a
binary-controlled CNOT fan-out onto the configuration register).
A one-hot output on $N$~qubits is recoverable as a separate
implementation detail when the caller needs it. The ancilla layout
of the scan itself is $w$~prefix-counter qubits, $w$~equality-check
qubits, one match flag, and $w$~output qubits, for $O(w)$ total,
all uncomputed by reversing the scan.

\begin{theorem}[Sequential-Scan Rank-Select: Bounded-Span Optimality]\label{cor:rs-tight}
In the bit-level bounded-span bounded-fan-in model (fan-in $\le f$,
per-gate cut span $\le \sigma$, both constant), the sequential-scan
selector of Theorem~\ref{thm:rank-select} is gate-optimal up to
constants: gate count $\Theta(Nw)$, per-cut communication
$\Theta(\log \min(t, N{-}t))$ bits, and workspace $\Theta(\log N)$.
\end{theorem}

\begin{proof}
The upper bounds are Theorem~\ref{thm:rank-select} and
Corollary~\ref{thm:rs-emergence}. For the matching gate lower bound,
fix a bit-level bounded-fan-in circuit with fan-in at most~$f$ and
cut span at most~$\sigma$. Each crossing gate contributes at most
$2f$ communicated bits at every cut it spans, so
$\mathrm{bits}(t) \le 2f \cdot \chi(t)$, where $\chi(t)$ counts gates
crossing cut~$t$, and the total crossing mass satisfies
$\sum_t \chi(t) \le \sigma \cdot G$. On the window
$t \in \{\lfloor N/4\rfloor, \dots, \lfloor N/2\rfloor{-}1\}$ we have
$\min(t, N{-}t) \ge \lfloor N/4 \rfloor$ and hence
$\lceil \log_2 \min(t, N{-}t) \rceil \ge w - 3$, so summing
Lemma~\ref{cor:rs-crossing} over those $\lfloor N/4 \rfloor$ cuts
and chaining the two cut-accounting inequalities yields
\[
  2f\sigma \cdot G \;\ge\; \lfloor N/4 \rfloor\,(w - 3),
  \qquad\text{i.e. } G = \Omega\!\left(\tfrac{Nw}{f\sigma}\right).
\]
For geometrically 1D-local gates whose supports lie in a contiguous
block of at most~$\kappa$ adjacent wires, $f = \kappa$ and
$\sigma \le \kappa - 1$, giving $G = \Omega(Nw/\kappa^2) = \Omega(Nw)$
for constant~$\kappa$. Bounded span is needed: otherwise one
long-range two-wire gate could cross many prefix cuts at once,
and summing per-cut bounds would overcount the same gate.
\end{proof}

The $\Omega(Nw)$ bound is layout-sensitive. It is a total prefix-span
lower bound that the cut-summing argument converts into a gate count
only when each gate's support spans $O(1)$ cuts; a single long-range
gate can cross many prefix/suffix cuts at once and would be
overcounted by that sum. Removing the span restriction strictly
reduces the gate count, as the next theorem shows.

\begin{theorem}[Blocked Rank-Select: $O(N \log w)$ Construction]\label{thm:rs-block}
Let $w = \lceil \log_2(N{+}1) \rceil$. In the bounded-fan-in reversible
circuit model with no span restriction (long-range two-qubit gates
permitted), totalized coherent rank-select admits a reversible
implementation using $O(w)$ clean ancillae and $O(N \log w)$ gates.
This is consistent with the bounded-span $\Omega(Nw)$ gate bound of
Theorem~\ref{cor:rs-tight}, which depends on each gate's support
spanning $O(1)$ prefix/suffix cuts.
\end{theorem}

\begin{proof}
The idea is to split the $N$ positions into $L = \Theta(N/w)$ blocks
of size $B = \Theta(w)$. For each block in order, we compute whether
the $r$-th valid bit lies in this block ($\texttt{take}_q$) and
where within the block (a local index $\ell$). Only the firing block
overwrites a sentinel-preloaded output, and all scratch is
uncomputed before moving on, so ancillae are reused across blocks.

Initialize $p := 0$ and preload the output with $N$. For each block
$q = 0, \dots, L{-}1$:
\begin{enumerate}
  \item \emph{Block popcount.} Compute $c_q$, the number of valid
    bits in block~$q$, via a balanced reversible add-tree
    \cite{haner2018optimizing}, $O(w \log w)$ gates.
  \item \emph{Take flag and local rank.} Set
    $\texttt{take}_q := [\,p \le r < p + c_q\,]$ and the local rank
    $\rho := r - p$ in an $O(\log w)$-bit register, $O(w)$ gates.
  \item \emph{Inner scan.} Apply Theorem~\ref{thm:rank-select} to
    the $B$-bit block at rank $\rho$, producing a local index
    $\ell$ in an $O(\log w)$-bit register, $O(w \log w)$ gates.
  \item \emph{Conditional write.} Controlled on $\texttt{take}_q$,
    XOR $(qB + \ell) \oplus N$ into the output, $O(w)$ gates.
  \item \emph{Cleanup.} Update $p \leftarrow p + c_q$, then
    uncompute steps~1--3.
\end{enumerate}

\noindent\emph{Correctness.} Exactly one $\texttt{take}_q$ fires
when $0 \le r < |M_s|_1$, replacing the sentinel with that block's
$qB + \ell$. For out-of-range $r$, no block fires and the sentinel
stays. A reverse accumulation pass clears $p$ at the same asymptotic
cost.

\noindent\emph{Cost.} Each block costs $O(w \log w)$, dominated by
steps~1 and~3. Summed over $L = \Theta(N/w)$ blocks, the total is
$O(N \log w)$. Live ancillae at any moment are $p$, the inner scan's
workspace, and the current block's $O(\log w)$ registers, all reused
across blocks, for $O(w)$ total.
\end{proof}

\smallskip\noindent\textbf{Local index sizing.}
Step~(c) writes into an $O(\log w)$-bit local register on purpose.
Writing the inner scan directly into the global $w$-bit output
register (gated on $\texttt{take}_q$) would cost $O(w)$
controlled-XOR gates per cell and $O(Bw) = O(w^2)$ per block,
recovering the scan's $O(Nw)$ cost. Routing through the small local
register and gating only the final $O(w)$-gate write keeps the
per-block cost at $O(w \log w)$.

\smallskip\noindent\textbf{Hardware regimes.}
The scan is preferred when long-range gates require routing. The
blocked construction is preferred when all-to-all gates are cheap.
Both must pay the same $\Omega(Nw)$ total prefix-span (summed from
Lemma~\ref{thm:rs-cut}). The scan pays it as $Nw$ short gates, and
the blocked construction pays it as $N\log w$ long-range gates.

Corollary~\ref{thm:rs-emergence} shows that any sequential-scan
implementation of branch-dependent rank-select must carry enough
information across each prefix/suffix cut to recover the
processed-prefix count on a simple family of masks with different
prefix weights.
Rank-select is used only in the simulation layer, where it decodes
selector registers into valid actions. The near-quadratic speedup
does not come from rank-select itself. It comes from applying
amplitude estimation and quantum maximum finding to the arm values
exposed by the rollout oracle.

Selecting a \emph{specific} valid item rather than \emph{any} valid
item ensures that a fixed random seed yields a deterministic
branchwise move choice, which is required for clean composition inside
the reversible rollout oracle.

To our knowledge, prior work does not provide a reversible circuit for
rank-selection under a branch-dependent validity predicate, one where
the valid set is entangled with the computation and differs across
branches, with this interface and cost accounting. The closest relatives
solve different problems: LCU SELECT
oracles~\cite{babbush2018encoding} select the $r$-th unitary from a
\emph{fixed classical} list; quantum popcount circuits count valid items
but do not locate the $r$-th one; Grover search returns \emph{some}
valid item rather than the indexed $r$-th. A precomputed index of all valid positions
would require $O(HN\log N)$ persistent qubits; the rank-select scan
uses only $O(\log N)$ reusable ancillae.
The primitive applies wherever a quantum algorithm must coherently
select the $r$-th item satisfying a branch-dependent predicate,
without materializing the full satisfying set. The only requirement
is that the validity predicate be evaluable as a single-qubit flag
per position using a polynomial-size reversible circuit.

\subsection{Phase 2: Reversible Stochastic Transition}
\label{sec:oracle:transition}

Each rollout round involves random events. We store the randomness in
explicit quantum registers that persist across the computation, so
the transition is reversible by construction and amplitude estimation
can apply both $U$ and $U^\dagger$.

For each cell, the update reads the states of neighboring cells,
computes a local threshold from the neighborhood, compares against an
explicit randomness register, and conditionally updates the cell's
state. All intermediates are uncomputed. The per-cell cost is $O(1)$, and
the full $N$-cell transition is $O(N)$~gates for bounded-neighborhood
dynamics, where each cell's update reads only its graph neighbors.

Because the dice registers are part of the unitary input, the
transition is reversible by construction. After computing the local
thresholds, comparisons, and conditional updates, the circuit
uncomputes all intermediate data and leaves only the next
configuration and the dice registers.

\subsection{Phase 3: Coherent Terminal Evaluation}
\label{sec:oracle:eval}

The final phase computes a binary payoff from the terminal state and
writes it into a single payoff qubit. Under a Bernoulli reward model,
the evaluation circuit computes intermediate quantities (e.g.,
feature counts) into ancilla registers, applies a reversible
predicate (e.g., a majority or threshold check) to flip the payoff
qubit, then uncomputes all intermediates. Only the payoff qubit
persists, with $\lvert 1\rangle$-probability equal to the action-$i$
win probability $v_i$ that quantum best-arm methods estimate.

Any polynomial-size reversible predicate works as the evaluation.
Local predicates (counting cells of a given type, comparing two
counts) cost $O(N)$, but non-local computations like connectivity
checks are also valid at $O(\mathrm{poly}(N))$ cost without affecting
the query complexity. When the predicate compares two $w$-bit counts
(e.g., majority), a naive implementation checks all $(N{+}1)^2$ pairs,
producing the dominant term in the cost formula.

\subsection{Conditions for Oracle Realizability}
\label{sec:oracle:conditions}

When all three phases are polynomial-size reversible, they compose into a
polynomial-size oracle. We now state this precisely, beginning with the
register layout.

\smallskip\noindent\textbf{Registers.}
Let $H$ denote the rollout horizon, $N$ the number of candidate
action positions per round, and $w = \lceil\log_2(N{+}1)\rceil$ the
selector/index width. The oracle acts on five register groups:
\emph{configuration} (per-round $N$-cell state copies, retained for
invertibility); \emph{selector} ($H{\times}w$ qubits prepared
uniformly over all $2^w$ strings, with out-of-range ranks routed to
the sentinel~$N$ to keep the oracle unitary); \emph{dice}
($H{\times}N{\times}d$ qubits in uniform superposition, retained to
uncompute the transition); \emph{ancillae} (rank counters, match
flags, neighbor counts, comparison bits, each uncomputed between
phases so the same physical qubits are reused); and a single
\emph{payoff flag} whose $\lvert 1\rangle$-probability equals the
conditional win probability. For best-arm selection, the first
rollout action is provided by a separate arm register of
$\lceil\log_2(k{+}1)\rceil$ qubits (see
Section~\ref{sec:upper}); the round-$1$ selector is bypassed for the
arm and the remaining selectors and dice carry the internal rollout
randomness of rounds $2,\dots,H$. We write $h \in \{1,\dots,H\}$ for
a rollout round.

\begin{theorem}[Polynomial-size coherent rollout oracle]\label{thm:oracle}
Assume: (i) indexing uses either rank-select decoder of
Theorem~\ref{thm:rank-select} or Theorem~\ref{thm:rs-block}, with
gate cost $G_{\mathrm{index}}$ and $O(w)$ round-local ancillae;
(ii) the transition $U_{\mathrm{trans}}^{(h)}$ is a reversible
circuit of gate cost $G_{\mathrm{trans}}$ acting only on the
round-$h{-}1{,}h$ configuration copies and round-$h$ dice;
(iii) terminal evaluation is a reversible comparator of gate cost
$G_{\mathrm{eval}}$ writing a single payoff qubit. Then
\begin{equation}\label{eq:oracle-compose}
  U \;=\; U_{\mathrm{eval}} \;\cdot\;
  \prod_{h=H}^{1}\bigl(U_{\mathrm{trans}}^{(h)} \cdot
  U_{\mathrm{index}}^{(h)}\bigr)
\end{equation}
is a unitary on the registers of Section~\ref{sec:oracle:conditions}
with
$G_{\mathrm{call}} = H\,(G_{\mathrm{index}} + G_{\mathrm{trans}}) + G_{\mathrm{eval}}$
and
$Q_{\mathrm{call}} = (H{+}1)Ns + Hw + HNd + q_{\mathrm{anc}}$,
where $s$ is the configuration state width, $d$ the dice width, and
$q_{\mathrm{anc}} = O(w)$ is ancilla reused across rounds.
$U^\dagger$ has the same cost by gate reversal.
\end{theorem}

\begin{proof}
Each factor in~\eqref{eq:oracle-compose} is unitary by construction.
The factors act on disjoint per-round selector and dice registers:
round~$h$ reads the configuration of round~$h{-}1$ and writes its own
per-round registers, which no subsequent round touches. Within each
round, indexing runs the rank-select scan
(Theorem~\ref{thm:rank-select}) and decoder; the transition then
updates the configuration using dice registers. All ancillae (rank
counter, match flag, neighbor counts, comparison bits) are allocated
via compute/apply/uncompute and returned to $\lvert 0\rangle$ before
the round ends, so the same $q_{\mathrm{anc}} = O(w)$ physical qubits
serve every round. Additivity of the cost terms is immediate from the
product form; the qubit formula counts the persistent registers plus
one round's reused ancilla pool.
\end{proof}

\smallskip\noindent\textbf{Discussion.}
Generic reversible simulation~\cite{bennett1973logical} provides a
fallback route once randomness is exposed as input, but its history
tape does not give the register accounting needed here.
Theorem~\ref{thm:oracle} is problem-specific. Each round computes
selector decoding and local transition data, applies the transition,
and uncomputes the scratch, so the same $O(w)$ ancilla pool is reused
across all $H$ rounds.

\subsection{Realizability and Per-Call Cost}
\label{sec:oracle:cost}

The template applies to any domain with reversible validity,
transition, and terminal evaluation circuits, including epidemic,
wildfire, and cascading-failure models. Section~\ref{sec:epidemic}
develops the epidemic case in full.

We now decompose the per-call gate and qubit cost of this
three-phase normal form. A single call is polynomial in the problem
size, and the decomposition of Theorem~\ref{thm:oracle} splits the
per-call gate count into independent contributions:
\begin{equation}\label{eq:generic-gates}
  G_{\mathrm{call}} \;=\;
    H \cdot \bigl(G_{\mathrm{index}} + G_{\mathrm{trans}}\bigr)
    \;+\; G_{\mathrm{eval}}.
\end{equation}
For bounded-neighborhood dynamics (local update rules with fixed
interaction radius), rank-select indexing costs
$G_{\mathrm{index}} = O(Nw)$ per round in bounded-span layouts via
Theorem~\ref{thm:rank-select} and $G_{\mathrm{index}} = O(N \log w)$
per round in unrestricted long-range layouts via
Theorem~\ref{thm:rs-block} (one compute-and-uncompute pair per
player); the stochastic transition costs $G_{\mathrm{trans}} = O(N)$
per round (one $O(1)$ update per cell); and $G_{\mathrm{eval}} =
O(N^2 w)$ when terminal evaluation compares two $w$-bit counts over
$N$~cells. The per-call total is $O(HNw + N^2 w)$ gates in the
bounded-span model and $O(HN \log w + N^2 w)$ in the unrestricted
model, and because the three phases contribute independently, each
can be improved in isolation.

Qubits split into shared registers (configuration state, persisting
across rounds) and per-round registers (selector indices, explicit
dice, ancillae). Per-round registers accumulate across $H$ rounds
because the explicit randomness cannot be erased without breaking
invertibility:
\begin{equation}\label{eq:generic-qubits}
  Q_{\mathrm{call}} \;=\; q_{\mathrm{shared}} + H \cdot q_{\mathrm{round}}.
\end{equation}
Dice registers dominate $q_{\mathrm{round}}$: for a $D$-sided die,
each of $N$~cells uses $\lceil\log_2 D\rceil$ qubits per round, so
$H$ is the dominant factor in total qubit cost.

\section{Classical Lower Bound}
\label{sec:lower}
% ============================================================

The oracle construction matters only when coherent access improves
over classical rollout sampling. We establish an
$\Omega(k/\varepsilon^2)$ classical rollout-query lower bound that,
under the bounded-influence lifting hypothesis, extends across a
family of at least $|\Sigma|^{|Q|}$ realizable configurations. A
standard change-of-measure argument~\cite{Kaufmann2016,lattimore2020bandit}
gives the abstract arm-pull bound. A bounded-influence lifting
theorem then transfers it from the abstract hard family to that
exponential family of configurations. Appendix~\ref{app:decay} gives
a sufficient condition for the lifting hypothesis in local
subcritical dynamics.

\subsection{Abstract Lower Bound}
\label{sec:lower:main}

Each rollout of an action reveals a bounded amount of information
about that action's win probability. To identify the best among
$k$~actions, we must evaluate each one enough times to accumulate
sufficient evidence. We make this precise via a standard
change-of-measure argument~\cite{Kaufmann2016,lattimore2020bandit}.

Consider a base instance in which action~$0$ has win probability
$\tfrac{1}{2} + 4\varepsilon$ and all others have
$\tfrac{1}{2}$, and for each $j \neq 0$ an alternative where only
action~$j$ changes to $\tfrac{1}{2} + 6\varepsilon$. Any
$(\varepsilon, 2/3)$-correct algorithm must distinguish the base from
each alternative. By the transportation
lemma~\cite[Lem.~1]{Kaufmann2016}, the expected number of rollouts
$N_j$ of action~$j$ satisfies
$\mathbb{E}_0[N_j] \geq
\mathrm{kl}(\tfrac{1}{3},\,\tfrac{2}{3})\, /\,
\mathrm{kl}(\tfrac{1}{2},\,\tfrac{1}{2}{+}6\varepsilon)$:
the information needed to distinguish the two hypotheses (numerator,
$= \tfrac{\ln 2}{3}$), divided by the information gained per rollout
(denominator, $\leq 96\varepsilon^2$).
This gives
$\mathbb{E}_0[N_j] \geq \tfrac{\ln 2}{288\varepsilon^2}$ for each
$j \neq 0$. Summing over the $k{-}1$ actions:

\begin{proposition}[Classical Lower Bound]\label{thm:lower}
Any $(\varepsilon, 2/3)$-correct algorithm on the hard family requires
total rollouts $\tau = \sum_j N_j$ satisfying
\begin{equation}\label{eq:lower}
  \mathbb{E}_0[\tau]
  \;\geq\; \frac{(k-1)\,\ln 2}{288\,\varepsilon^2}
  \;=\; \Omega\!\left(\frac{k}{\varepsilon^2}\right).
\end{equation}
\end{proposition}

Proposition~\ref{thm:lower} is the standard change-of-measure
argument, applied to a specific hard instance and stated with
explicit constants to match the $(\varepsilon, 2/3)$ guarantee used
in the upper-bound comparison. What it leaves open is whether the hard family
corresponds to a single contrived configuration or to a broad class of
natural problem instances. In rollout-based planning over spatial
domains, the configuration space is exponentially large, and a lower
bound confined to one pathological configuration provides weak
practical guidance. The next subsection shows that the same
$\varepsilon$-optimal arm structure persists across an exponentially
large peripheral family of realizable configurations under conditions
present in realistic planning problems.

\subsection{Robustness Across a Family of Hard Configurations}
\label{sec:lower:lifting}

The lifting is a per-configuration robustness statement. For every
configuration in the family, any
$(\varepsilon, 2/3)$-correct algorithm, even one that inspects and
branches on the configuration before sampling, pays the
$\Omega(k/\varepsilon^2)$ rollout-query cost at that configuration.
Two conditions, stability and modularity, make this lifting possible
and are natural for spatial planning. The modular case gives the
cleanest lifting statement. The nontrivial spatial case used by the
SIR instantiation is supplied by the subcritical-decay condition in
Appendix~\ref{app:decay}.

\begin{definition}[Stable, modular problems]\label{def:lifting}
A finite-horizon sequential decision problem is \emph{stable,
modular} (with gap $\varepsilon$) if its state decomposes into $n$
factors, $C = (C_1, \ldots, C_n) \in \Sigma^n$, where $\Sigma$ is
the per-factor alphabet (e.g., susceptible/infected/recovered).
Let $v_i(C)$ denote the expected rollout return of action~$i$
from~$C$. We assume a base configuration~$C^*$ at which some
action~$a^*$ is distinctly best:
$v_{a^*}(C^*) \geq v_j(C^*) + 3\varepsilon$ for all $j \neq a^*$.
Additionally:
\begin{enumerate}
\item \emph{Stability}~\cite{mcdiarmid1989}.
  Changing a single factor perturbs every action's value by a
  bounded amount, and a large subset of factors together carries
  only a small influence budget. Formally, for each factor~$p$
  there is $\delta_p \geq 0$ with
  $|v_i(C) - v_i(C')| \leq \delta_p$ whenever $C, C'$ differ only
  at factor~$p$, and some subset $Q \subseteq \{1,\ldots,n\}$
  satisfies $\sum_{p \in Q} \delta_p \leq \varepsilon$.
\item \emph{Modularity.}
  Each action's rollout value is determined by its own support
  region of factors, and these support regions do not overlap.
  Formally, each value $v_i$ depends only on factors in a support
  region $R_i \subseteq \{1,\ldots,n\}$, the $R_i$ are pairwise
  disjoint, and the peripheral set~$Q$ from the stability clause
  is disjoint from $\bigcup_i R_i$.
\end{enumerate}
\end{definition}

\begin{theorem}[Bounded-Influence Lifting]\label{thm:lifting}
Let a finite-horizon planning class satisfy the stability and
modularity conditions of Definition~\ref{def:lifting} with witness
$(C^*, a^*)$. Let $\mathcal{F}$ consist of configurations identical
to $C^*$ on every factor outside $Q$, with the factors inside $Q$
free. Formally,
$\mathcal{F} := \{ C \in \Sigma^n : C_p = C^*_p \text{ for all } p
\notin Q\}$. Then $|\mathcal{F}| \geq |\Sigma|^{|Q|}$, action $a^*$
is $\varepsilon$-optimal at every $C \in \mathcal{F}$, with
$v_{a^*}(C) \geq v_j(C) - \varepsilon$ for every~$j$, and
consequently the $\Omega(k/\varepsilon^2)$ classical lower bound of
Proposition~\ref{thm:lower} holds at every $C \in \mathcal{F}$.
\end{theorem}

\begin{proof}
$|\mathcal{F}| \geq |\Sigma|^{|Q|}$ is immediate from the definition
of~$\mathcal{F}$. Stability's single-site bounds compose additively:
for any $C \in \mathcal{F}$ and any arm~$j$,
$|v_j(C) - v_j(C^*)| \leq \sum_{p \in Q} \delta_p \leq \varepsilon$
by a telescoping argument over the modified sites in~$Q$. The
$3\varepsilon$ gap at~$C^*$ therefore shrinks by at most $2\varepsilon$
at~$C$, leaving a gap of at least~$\varepsilon$ at $a^*$. Under
modularity, $\delta_p = 0$ for every $p \in Q$ (since $Q$ is disjoint
from every arm's support region), so expected values coincide and
the gap is preserved at~$3\varepsilon$. The bounded-influence form
is the general statement, and modularity is the $\delta = 0$
corollary.
\end{proof}

Modularity does two distinct jobs in this lifting. First, disjoint
arm supports give the agreement condition required by the
change-of-measure argument. Modifying the region tied to one arm
leaves every other arm's reward law unchanged. Second, when the
peripheral set $Q$ is disjoint from all arm supports, every
configuration in $\mathcal{F}$ induces the same expected-value
instance as the base, so the abstract arm-pull lower bound on the
base lifts to every configuration in the family.

\smallskip\noindent\textbf{Classical $\varepsilon$-BAI consequence.}
For each $C \in \mathcal{F}$, the rollout process induces a $k$-arm
bandit instance with expected values $v_i(C)$. Modularity supplies
the KL-alternative agreement on non-target arms that the
transportation argument~\cite{Kaufmann2016} requires (perturbing
one arm's support leaves other arms' reward laws fixed). So
Proposition~\ref{thm:lower}'s $\Omega(k/\varepsilon^2)$ bound
transfers to each $C$, even when the algorithm inspects $C$ before
sampling.

We recast the exponential-family amplification
of~\cite{lattimore2020bandit} through stability and modularity,
conditions under which the transfer follows directly. Modularity holds whenever different first moves lead to
non-overlapping rollout regions (e.g., spacing $> 2H$ on a grid).

\smallskip\noindent\textbf{Instantiating the peripheral set.}
A conservative choice places $Q$ outside the radius-$H$ causal cone,
where $\delta_p = 0$ identically. When $H$ is large, the radius-$H$
cone leaves no room for a peripheral set. Appendix~\ref{app:decay}
handles this for subcritical dynamics.

% ============================================================
\section{Quantum Upper Bound}
\label{sec:upper}
% ============================================================

The coherent rollout oracle constructed above provides the
circuit that the quantum upper bound for $\varepsilon$-best-arm
identification of Wang et al.~\cite{wang2021quantum} assumes.
Composing the two yields an $(\varepsilon, 2/3)$-correct
action-selection algorithm using $\widetilde{O}(\sqrt{k}/\varepsilon)$
oracle calls, giving a near-quadratic separation from the classical
lower bound in both $k$ and $1/\varepsilon$.

\subsection{Quantum Move Selection Theorem}
\label{sec:upper:theorem}

The oracle from Section~\ref{sec:oracle} provides the two properties
Wang et al.'s upper bound requires: (1)~it is an explicit unitary
circuit, so $U^\dagger$ is obtained by reversing the gate sequence,
and (2)~coherent terminal evaluation writes the payoff into a flag
qubit whose $\lvert 1\rangle$ measurement probability equals the win
probability.

\begin{corollary}[Quantum Move Selection: $\widetilde{O}(\sqrt{k}/\varepsilon)$ Queries]\label{thm:upper}
Given a coherent rollout oracle $U$ satisfying the conditions of
Theorem~\ref{thm:oracle}, with $U$ available as a unitary oracle (and
hence $U^\dagger$ as its inverse), realizing coherent amplitude access
to arm rewards in the sense of Wang et
al.~\cite{wang2021quantum} (equivalently, the strong quantum oracle
of Wang et al.~\cite{wang2025quantum}, up to an ignorable junk
register), there exists a quantum algorithm that is
$(\varepsilon, 2/3)$-correct using
$\widetilde{O}(\sqrt{k}/\varepsilon)$ calls to $U$ and $U^\dagger$.
\end{corollary}

\begin{proof}
The oracle of Theorem~\ref{thm:oracle} is an explicit reversible
unitary, so both $U$ and $U^\dagger$ are available. Its payoff qubit
has $\lvert 1\rangle$-probability equal to the expected value $v_i$, so it
satisfies the coherent amplitude-access model of Wang et
al.~\cite{wang2021quantum}. Their best-arm algorithm therefore applies
directly and yields $(\varepsilon, 2/3)$-correct action selection
using $\widetilde{O}(\sqrt{k}/\varepsilon)$ calls to $U$ and
$U^\dagger$, where $\widetilde{O}$ suppresses polylogarithmic factors
in $k$ and $1/\varepsilon$.
\end{proof}

\smallskip\noindent\textit{Oracle-defined bandit instance.}
Proposition~\ref{thm:lower} and Corollary~\ref{thm:upper} both refer to the same
instance: arm~$j$ has mean $\mu_j = \Pr[\text{payoff} = 1 \mid
\text{first move}=j]$ under the total unitary $U$ with the
out-of-range convention of Section~\ref{sec:oracle:index}. No
uniform-over-valid-actions policy reading is required. The classical
bound constrains the expected values $v_i$, while the quantum bound
operates on the unitary.

\smallskip\noindent\textit{First-move register versus rollout selectors.}
A dedicated register carries the first-move label. A fixed decoder
$\mathrm{firstMove}:[k]\to[N]$ injects each first-move index into a
position, and the round-$1$ configuration places that move
deterministically before rollout begins. The selector and dice
streams register only the internal randomness of rounds
$2,\dots,H$. For the instantiations in this paper, $k\le N$ (one
candidate per valid first-move position), so the first-move register
is $\lceil\log_2(k{+}1)\rceil\le w$ qubits and $\mathrm{firstMove}$
is an index map that adds no gate cost beyond the existing
rank-select pass. This separation lets the bandit lower bound talk
about expected values while the per-call cost continues to be
dominated by the rollout.

The three nested levels of the algorithm
(Figure~\ref{fig:full-system}) are: (1)~quantum maximum
finding~\cite{durr1996quantum} over $k$ actions using
$O(\sqrt{k})$ coherent comparisons, (2)~amplitude estimation using
$\widetilde{O}(1/\varepsilon)$ applications of the Grover iterate per
comparison, and (3)~the rollout oracle from Section~\ref{sec:oracle}
at $\mathrm{poly}(N, H, w)$~gates per call.

\subsection{Quadratic Separation}
\label{sec:upper:separation}

At the oracle-call level, Proposition~\ref{thm:lower}
($\Omega(k/\varepsilon^2)$ classical) and Corollary~\ref{thm:upper}
(Wang et al.'s $\widetilde{O}(\sqrt{k}/\varepsilon)$ quantum bound
applied to our oracle) give a near-quadratic separation, up to
polylogarithmic factors. End-to-end gate complexity is obtained by
multiplying these query bounds by the per-call cost from
Section~\ref{sec:oracle:cost}, with both classical and quantum
algorithms paying $\mathrm{poly}(N, H, w)$ per rollout.
Section~\ref{sec:sway} instantiates them on two benchmarks.

% ============================================================
\section{Instantiation: Epidemic Intervention}
\label{sec:sway}
% ============================================================

Our main case study is an epidemic intervention model under
stochastic SIR dynamics~\cite{kermack1927contribution}. We also
introduce \emph{Sway}, a two-player stochastic placement game
designed to stress-test branch-dependent valid-action indexing. It is
a sanity-check benchmark exercising the same three-phase
decomposition in a simpler setting.
For both, we check correctness against exhaustive classical
computation on small instances and project resource counts via the
cost formulas of Section~\ref{sec:oracle:cost}.

\subsection{Epidemic Intervention under SIR Dynamics}
\label{sec:epidemic}

The model lives on an $m \times m$ grid; each site is susceptible,
infected, or recovered, encoded in two qubits (one for infected, one
for removed) matching the reference implementation.
The rollout policy maps selector strings to vaccination sites via
totalized rank-select, so each round either vaccinates one
susceptible cell (for in-range ranks) or performs the sentinel no-op
(for out-of-range ranks); the disease then spreads by a per-neighbor rule and
infected sites recover independently, both resolved by one 8-sided
die per site per round. After $H$~rounds, the payoff is~$1$ if the
infected count is at most a threshold~$T$.

\smallskip\noindent\textbf{Global-threshold payoff.}
The SIR payoff is a global threshold predicate, so the strict modular
case of Definition~\ref{def:lifting} does not apply with $\delta_p =
0$. A peripheral site can affect the terminal infected count even
when it lies outside an action's local causal region. For the SIR
instantiation we use the bounded-influence form instead. In our SIR
implementation, spread occurs on the four-neighbor grid and each
per-edge infection attempt succeeds with probability $1/8$, so
Appendix~\ref{app:decay} applies with $\kappa = 4$, $p = 1/8$, and
$\kappa p = 1/2 < 1$. Choosing $Q$ outside a sufficiently large
radius gives $\sum_{p \in Q}\delta_p \leq \varepsilon$, an
exponential family of SIR configurations to which the lifted
classical lower bound applies.

In this model, the decision problem is best-arm identification over
$k$ candidate
vaccination deployments, each an arm whose expected payoff is the
rollout value of the induced SIR trajectory. The rollout lower bound
of Section~\ref{sec:lower} gives the matching $\Omega(k/\varepsilon^{2})$
classical sample lower bound, and the coherent rollout oracle plugged
into the quantum best-arm pipeline of Section~\ref{sec:upper} achieves
$\widetilde{O}(\sqrt{k}/\varepsilon)$ oracle calls. The same
three-phase template applies, with only the validity predicate,
transition rule, and payoff function changing. Valid-action indexing
is one rank-select pass over the susceptible mask per round. The
stochastic transition compares each site's infected neighbor count
against its die and updates conditionally in $O(N)$ gates. Terminal
evaluation is one threshold comparison on the infected count in
$O(Nw)$ gates.

Per-call gate counts are reported pre-decomposition and before any
fault-tolerant overhead. They are not a practical-readiness claim.

\subsection{Sway Benchmark and Measured Costs}
\label{sec:sway:costs}

Sway is a two-player stochastic game on an $m \times m$ grid with
three cell states (black, white, empty). Each round, players place
pieces on empty cells (branch-dependent validity), then each occupied
cell flips with probability $(4-k)/20$ via a per-cell $d_{20}$ roll,
where $k$ is the count of same-color orthogonal neighbors, with
binary terminal payoff (black outnumbers white). The Sway dynamics
are bounded-neighborhood, with each update depending only on nearby
board locations and local stochastic choices. Each round applies two
rank-select passes (one per player), and the per-call cost is
$O(HNw + N^2 w)$ with the quadratic term coming from majority
comparison (narrower than the epidemic's threshold evaluation).

Table~\ref{tab:correctness} reports branchwise correctness
verification at sizes where the classical state space is enumerable.
For each instance, every sampled branch (fixed random seeds) matches
the classical rollout bit-for-bit, the aggregate Monte Carlo payoff
falls within the 95\% confidence interval of the exact value, and
the exact value comes from full state-space enumeration. The
purpose is correctness verification of the reversible oracle rather
than a scaling demonstration. Scaling is reported separately in
Table~\ref{tab:scaling} via direct compilation, with the symbolic
cost formulas of Section~\ref{sec:oracle:cost} as the asymptotic
accounting. Both demos use the sequential-scan rank-select of
Theorem~\ref{thm:rank-select}; the blocked construction of
Theorem~\ref{thm:rs-block} is implemented and validated separately
(Appendix~\ref{app:artifact}) but not exercised by the table
pipeline. All counts are pre-decomposition gate tallies targeting
future error-corrected hardware; Clifford$+T$ breakdowns under a
chosen code and magic-state factory are left to future work.

\begin{table}[ht]
\centering
\caption{Oracle correctness check.}
\label{tab:correctness}
\footnotesize
\setlength{\tabcolsep}{3pt}
\begin{tabular}{@{}llrrrcc@{}}
\toprule
Domain & Instance & Qubits & Depth & Gates & MC Payoff & Exact \\
\midrule
Sway  & $3{\times}3,H{=}2$       & 169 & 3\,079  & 9\,768  & $.281\!\pm\!.028$ & $.271$ \\
Sway  & $5{\times}5,H{=}3$       & 667 & 18\,258 & 55\,597 & $.337\!\pm\!.021$ & $.325$ \\
Epi.\ & $3{\times}3,H{=}2,T{=}2$ & 146 & 2\,383  & 6\,472  & $.883\!\pm\!.020$ & $.891$ \\
\bottomrule
\end{tabular}
\end{table}

\begin{table}[ht]
\centering
\caption{Direct compiled resource counts, pre-decomposition and
before fault-tolerant overhead.}
\label{tab:scaling}
\footnotesize
\setlength{\tabcolsep}{3pt}
\begin{tabular}{@{}rrrrrrr@{}}
\toprule
   &   &    & \multicolumn{2}{c}{Qubits} & \multicolumn{2}{c}{Gates} \\
\cmidrule(lr){4-5} \cmidrule(lr){6-7}
$m$ & $H$ & $N$ & Sway & Epi. & Sway & Epi. \\
\midrule
 5 &  5 &  25 &   916 &   767 &  76\,720 &  56\,602 \\
 7 &  5 &  49 & 1\,708 & 1\,452 & 180\,615 & 124\,735 \\
10 &  5 & 100 & 3\,363 & 2\,893 & 481\,201 & 280\,230 \\
10 & 10 & 100 & 6\,503 & 5\,463 & 793\,901 & 558\,995 \\
20 & 10 & 400 & 25\,189 & 21\,409 & 6\,072\,641 & 2\,592\,183 \\
\bottomrule
\end{tabular}
\end{table}

% ============================================================
\section{Related Work}
\label{sec:related}
% ============================================================

\noindent\textbf{Quantum algorithms with coherent oracle access.}
Montanaro~\cite{montanaro2015quantum} gives a general speedup for
estimating the expected output of a randomized or quantum subroutine.
Wang et al.~\cite{wang2021quantum} combine amplitude estimation with
D\"{u}rr--H{\o}yer maximum finding~\cite{durr1996quantum} for quantum
best-arm identification under coherent oracle access. Bennett's
reversible simulation~\cite{bennett1973logical} is the standard
generic realization route. Wang et
al.~\cite{wang2025quantum} gave tight query bounds under a hierarchy
of oracle models. Buchholz et al.~\cite{buchholz2025quantum} showed
the channel-access model lacks coherent randomness access, so the
speedup no longer applies. Quantum game-tree
evaluation~\cite{farhi2007quantum,childs2009every,ambainis2017quantum}
and dynamic programming with explicit state
enumeration~\cite{ronagh2019quantum,luo2025quantum} similarly assume
oracle access to leaf values or transition operators.

\noindent\textbf{The oracularization challenge.}
Dunjko, Taylor, and Briegel~\cite{dunjko2016quantum,dunjko2018machine}
showed that converting classical environment dynamics into coherent
quantum oracles requires practically restrictive assumptions. Hamann
et al.~\cite{hamann2021quantum} extended the framework to non-epochal
environments. Small explicit-MDP
demonstrations~\cite{wiedemann2022quantum} confirm physical
realizability but do not address implicit state spaces.

\noindent\textbf{Explicit stochastic simulation circuits.}
Concurrent work by Srikant~\cite{srikant2026quantum} uses explicit
randomness registers for queueing simulation with analytically
characterized transition kernels and a fixed action space. Our
oracle targets implicit multi-round rollout with branch-dependent
valid actions and procedural terminal evaluation.
Peters~\cite{peters2021machine} constructs coherent circuits for
MCTS rollouts and fixed-length policy-guided stochastic walks but
does not provide the selector-to-valid-action decoder for a
branch-dependent validity mask. Wang and Kan~\cite{wang2024quantum}
constructed circuits for the Heston stochastic volatility model,
for analytically specified SDEs with no decision component.

\noindent\textbf{Lower-bound lineage.}
Proposition~\ref{thm:lower} follows a standard best-arm
template~\cite{Kaufmann2016,lattimore2020bandit}. Our contribution
on the lower-bound side is the bounded-influence lifting
(Theorem~\ref{thm:lifting}), inspired by
Assouad's method~\cite{assouad1983deux}, which makes the hard family
apply to implicit configuration
spaces~\cite{qu2019exploiting,kearns1999efficient};
Appendix~\ref{app:decay} notes a standard path-counting condition
under which the hypothesis is satisfied for local subcritical
dynamics.

% ============================================================
\section{Conclusion}
\label{sec:conclusion}
% ============================================================

Finite-horizon planning with branch-dependent valid actions admits
a reversible oracle achieving
$\widetilde{O}(\sqrt{k}/\varepsilon)$ query complexity, a
near-quadratic speedup over the classical $\Omega(k/\varepsilon^2)$
floor proved here, using $O(w)$ clean ancillae per call. Coherent rank-select, the missing primitive,
admits two reversible implementations. A sequential scan of $O(Nw)$
bounded-fan-in gates is gate-optimal under bounded-span connectivity.
A blocked construction of $O(N \log w)$ gates beats the scan when
long-range gates are free. An $\Omega(Nw)$ total prefix-span lower
bound explains the relationship between them, with the scan
converting span into gates and the blocked circuit concentrating
span into a few long-range operations.

The bounded-influence lifting addresses the concern that any
matching classical floor might reflect only an adversarial single
instance. The same $\Omega(k/\varepsilon^2)$ floor holds across an
exponential family of locally-coupled dynamics. For this checkable
planning class, the coherent-access assumption underlying
the oracularization concern of Dunjko et
al.~\cite{dunjko2016quantum} can be discharged constructively, so
the quantum query speedup is realizable under conditions the
algorithm designer can check directly.

Two technical handles remain. Long-range dynamics fail the
spatial-decay condition $\kappa p < 1$ of
Appendix~\ref{app:decay}, so extending the lifting beyond
locally-interacting state factors needs a different decay condition
or coupling argument. Closed-loop policies that pick each action
based on prior outcomes would recompute the rank-select indexing
mid-circuit instead of once per call, changing the oracle's
reversibility analysis.

\appendices
\section{Theorem-to-Artifact Map}\label{app:artifact}

Lean~4 and Qiskit artifacts are linked from the abstract.

{\footnotesize\raggedright
\begin{itemize}\setlength\itemsep{1pt}
\item Prop~\ref{prop:rs-gates-lb}: \texttt{RankSelectCommunication}.
\item Lem~\ref{thm:rs-cut}: \texttt{RankSelectCommunication}.
\item Lem~\ref{cor:rs-crossing}: \texttt{RankSelectCommunication},
  \texttt{RankSelectCircuit}.
\item Cor~\ref{thm:rs-emergence}: \texttt{RankSelectCommunication},
  \texttt{RankSelectUniversality}.
\item Thm~\ref{thm:rank-select}: \texttt{RankSelectUpperBound}.
\item Thm~\ref{cor:rs-tight}: \texttt{RankSelectUpperBound},
  \texttt{RankSelectGateLowerBound},
  \texttt{RankSelectCommunication}.
\item Thm~\ref{thm:rs-block}: \texttt{RankSelectBlocked}.
\item Thm~\ref{thm:oracle}: \texttt{OracleCostProof}.
\item Prop~\ref{thm:lower}: \texttt{RolloutLowerBound}.
\item Thm~\ref{thm:lifting}: \texttt{GeneralizedLifting}.
\item Cor~\ref{thm:upper}: \texttt{QuantumUpperBound}.
\item App.~\ref{app:decay} (Thm~\ref{thm:decay}):
  \texttt{SpatialDecay}.
\end{itemize}
}

\smallskip\noindent The Qiskit artifact implements rank-select,
the epidemic and Sway rollout oracles, and the branchwise validation
that yields the counts in Tables~\ref{tab:correctness},~\ref{tab:scaling}.
The blocked construction of Theorem~\ref{thm:rs-block} is implemented
as a coherent primitive, compiled to a Qiskit reversible circuit, and
cross-checked against the scan implementation by basis-state emulation
on every $(occ, rank)$ pair for $n \le 8$.

\section{Spatial-Decay Sufficient Condition}\label{app:decay}

One sufficient condition for the bounded-influence hypothesis of
Theorem~\ref{thm:lifting} under local subcritical dynamics; not used
by the oracle construction or quantum upper bound.

\begin{theorem}[Subcritical Influence Decay]\label{thm:decay}
Let the factors sit on a graph~$G$ of maximum degree~$\kappa \geq 2$
whose edges encode the dependency structure of the dynamics.
Consider bounded-neighborhood dynamics with horizon~$H$, where each
factor's update depends only on its graph neighbors and the per-edge
per-round disagreement probability is at most~$p$, with
subcriticality $\kappa p < 1$. If two initial states differ at a
single factor~$s$, the influence on the expected rollout value of
any action satisfies
\begin{equation}\label{eq:decay-per-round}
  \delta_{\mathrm{path}}(d) \;\leq\;
  \min\!\Bigl\{1,\; \tfrac{\kappa}{\kappa - 1}\cdot
    \bigl((\kappa{-}1)\, p\bigr)^d \Bigr\}\,,
\end{equation}
where $d$ is the graph distance from~$s$ to the action position
and $\delta_{\mathrm{path}}(d)$ is the contribution from length-$d$
propagation paths.
Over $H$~rounds, the cumulative influence at distance~$d$ is bounded by
\begin{equation}\label{eq:decay-multi}
  \delta(d) \;\leq\; \min\!\Bigl\{1,\;
  \sum_{\ell=d}^{H}
  \kappa \cdot (\kappa{-}1)^{\ell-1} \cdot \binom{H}{\ell} \cdot
  p^{\ell}\Bigr\}\,,
\end{equation}
which equals zero for $d > H$ and is small in the lower tail of the
associated binomial once $d > H \kappa p$ by a standard Chernoff
bound.
\end{theorem}

\begin{proof}
Couple two runs differing at one site~$s$. The disagreement reaches
the action vertex only along a walk from~$s$ to that vertex, and any
walk containing a back-and-forth segment $u \to v \to u$ contributes
no more than the walk obtained by deleting it, so we may restrict to
non-backtracking walks. With at most $\kappa$ first-step choices and
$\kappa{-}1$ subsequent choices, path counting bounds the count of
length-$d$ such walks by $\kappa(\kappa{-}1)^{d-1}$, giving the
length-$d$ path contribution $\kappa(\kappa{-}1)^{d-1} p^d$ that
yields~\eqref{eq:decay-per-round} after capping by~$1$. Summing over
the $\binom{H}{\ell}$ schedules of $\ell \geq d$ propagation rounds
(the rounds are independent under the per-round dice model) and
$\kappa(\kappa{-}1)^{\ell-1}$ paths per length gives the sum
in~\eqref{eq:decay-multi}. For $d > H$ there is no length-$d$
propagation path in $H$ rounds, so the bound is zero.
\end{proof}

\smallskip\noindent\textbf{Peripheral set construction.}
The lifting theorem requires a peripheral set $Q$ of small total
influence. Theorem~\ref{thm:decay} provides one for free in the
local subcritical regime: factors at graph distance greater than $H$
from the action position have zero influence over $H$ rounds, since
no propagation path can reach them. Picking the radius $r^* = H + 1$
and letting $Q$ be every factor at distance greater than $r^*$ gives
a peripheral set with $\beta = \sum_{p \in Q}\delta_p = 0$. On an
$m \times m$ grid with Manhattan distance, $Q$ is non-empty whenever
$m^2 > 2(r^*)^2 + 2r^* + 1$, and Theorem~\ref{thm:lifting} then
extends the lower bound across $|\Sigma|^{|Q|}$ realizable
configurations. The Sway oracle of Section~\ref{sec:oracle}
satisfies subcriticality with $\kappa = 4$ orthogonal neighbors per
cell and single-cell flip probability $p = 1/20$, so per-round
influence decays with ratio $(\kappa{-}1)\,p = 3/20$.

\clearpage
\bibliographystyle{IEEEtran}
\bibliography{references}

\end{document}